\documentclass[reqno]{amsart}

\usepackage{amssymb,amsmath,amscd,amsthm}
\usepackage{cite}
\usepackage{graphicx}
\graphicspath{ {./images/} }
\usepackage{float} 
\usepackage{graphicx,psfrag}

\newtheorem{theorem}{Theorem}[section]

\newtheorem{lemma}[theorem]{Lemma}

\date{}
\begin{document}
\title{Beta Critical for the Schr\"odinger Operator with Delta Potential }
   \author{
 Rajan Puri }
 
 \thanks{Department
 of Mathematics and Statistics, Wake Forest University,
Winston Salem, NC 27109, USA, (purir@wfu.edu).  }

\date{}


\begin{abstract} 
  For the one dimensional Schr\"odinger operator in the case of Dirichlet boundary condition, we show that $\beta_{cr}$ is positive and zero for the case of Neumann and Robin boundary condition considering the potential energy of the form $V(x)=-\beta \delta(x-a)$ where, $\beta \geq 0, \ a > 0.$ We prove that the $\beta_{cr}$ goes to infinity when the delta potential moves towards the boundary in dimension one with Dirichlet boundary condition. We also show that the $\beta_{cr}>0$ and $\beta \in (0,\frac{1}{2})$ considering Dirichlet problem with delta potential on the circle in dimension two. 
\end{abstract}



\maketitle

\section{Introduction} This paper is a continuation of our work \cite{PV} where we studied the beta critical of the coupling constant of exterior elliptic problems and proved that $\beta_{cr}>0$ in the case of Dirichlet boundary condition and $\beta_{cr}=0$ in the case of the Neumann boundary condition in the dimension d=1 \text{and} 2 by studying the truncated resolvent operator. It was shown that the choice $\beta_{cr}>0$ or $\beta_{cr}=0$  depends on whether the truncated resolvent is bounded or goes to infinity when $\lambda\to 0^-$. In fact, $\beta_{cr}$ was expressed through truncated resolvent operator. There has been a considerable interest in the study of coupling constant and the problem was investigated by  several researchers like Barry Simon in\cite{BS}, Martin Klaus in \cite{MK}, Cranston, Koralov, Molchanov and Vainberg in \cite{CKMV}, and Yuriy Golovaty in \cite{YG}. It is known that the spectrum of $-\Delta-\beta V(x)$ consists of the absolutely continuous part $[0, \infty)$ and at most a finite number of of negative eigenvalues.
    $$\sigma(-\Delta-\beta V(x)) = \{\lambda_j\} \cup [0, \infty), \ 0\leq j\leq N , \quad \lambda_j \leq 0.$$
    We proved the dependence of beta critical with the boundary condition and dimension in our earlier paper \cite{PV}. Namely, 
\begin{theorem}[\cite{PV}]\label{tmr}
Consider the following elliptic problems in $\Omega$.
    \begin{equation}\label{eq:10}
        H_{0}u - \beta V(x)u -\lambda u= f , \ \ x\in \Omega,
    \end{equation}
    where $H_{0}= -\text{div}(a(x)\nabla)$, the potential $V(x)\geq 0$ is compactly supported and continuous, $\beta\geq0$ , $a(x)> 0 \ , \ a(x)\in C^1(\Omega)$, and $a=1 \ \text{when} \ |x|>>1$.
 If  $d=1$ or $ 2$ then $\beta_{cr}>0$ in the case of the Dirichlet boundary condition, and $\beta_{cr}=0$ in the cases of the Neumann boundary condition. 
\end{theorem}
In \cite{PV}, we also studied the dependence of $\beta_{cr}$ on the distance between the support of the potential and the boundary of the domain. In fact, it was proven that, in the case of the Dirichlet boundary condition in the dimension one, $\beta_{cr}$ tends to infinity as potential moves towards the boundary.  In dimension two with the Dirichlet boundary condition, the behavior of $\beta_{cr}$ was interesting and depends on the relation between the rates of the shrinking of the support of the potential and the speed of its motion towards the boundary. We did not consider the Neumann boundary condition when $d=1$ or $2$ since $\beta_{cr}$ is always zero. In particular, we proved the following theorem in \cite{PV}. 
\begin{theorem}[\cite{PV}]\label{tlast}
If $d=1$, then $\beta_{cr}$ for the Schr\"{o}dinger operator
$-\Delta-\beta V(x)$ goes to infinity as $n\to\infty$ for the Dirichlet boundary condition. The same is true if $d=2$ and $|x(n)-x_0|<C/n,~n\to\infty$. If $d=2$ and $|x(n)-x_0|\to 0,~|x(n)-x_0|>C/n^\delta,~n\to\infty,$ with some $\delta\in(0,1)$, then $\beta_{cr}$ remains bounded as $n\to\infty$. If $d\geq 3$, then $\beta_{cr}$ remains bounded as $n\to\infty$ for both the Dirichlet and Neumann boundary conditions.
\end{theorem}
\section{Schr\"odinger operator with Delta Potential}
In this paper, we would present results on the beta critical in the case of one and two dimensional Schr\"odinger equation with delta potential given by
\begin{equation}\label{eq0}
    -y^{''}-\beta \delta(x-a) y(x)=\lambda y(x),
\end{equation}
where $\beta \geq 0, \ a >0$.  The delta function is a infinitely high, infinitesimally narrow spike at $x=a$.  This allows solutions for both the bound states $\lambda <0$ and scattering states $\lambda >0.$  The classification of the spectrum into discrete and continuous parts usually corresponds to a classification of the dynamics into localized (bound) states and locally decaying states when time increases (scattering), respectively. The lower bound, $0$, of the absolutely continuous spectrum  is called the ionization threshold. This follows from the fact that the particle is no longer localized, but moves freely when $\lambda>0$. This classification is related to the space-time behaviour of solutions of the corresponding Schr\"odinger equation.\\We are interested to study the beta critical, the critical value the coupling constant denoted by $\beta_{cr},$ the value of $\beta$ such that equation (\ref{eq0}) does not have negative eigenvalues for $\beta<\beta_{cr}$ and has them if $\beta>\beta_{cr}$. We can find the solution to Schrodinger equation (\ref{eq0}) in the region I and II as shown in the figure 1. 

\begin{figure}[h!]
    \centering
    \includegraphics[width=9cm, height=4cm]{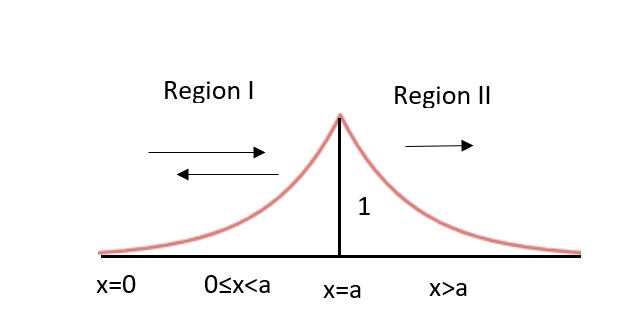}
    \caption{Delta Potential at x=a.}
    \label{fig:my_label}
\end{figure}
In both regions: region I or  $0\leq x<a$ and region 2 or $x>a$, the potential is $V(x)=0.$ 
\begin{theorem}\label{thm1}
If we consider the following one dimensional Schr\"odinger equation in the half axis 
\begin{equation}{\label{eq1}}
-y^{''}-\beta \delta(x-a) y(x)=\lambda y(x)
, \ \lambda= -k^2 < 0 \ x\in[0,\infty).
\end{equation}
 then the $\beta_{cr}>0$ in the case of Dirichlet Boundary condition and $\beta_{cr}=0$ in the case of Neumann and Robin boundary condition. 
\end{theorem}
\begin{proof}

The Dirichlet problem is given by
\begin{equation}{\label{eq4}}
-y^{''}-\beta \delta(x-a) y(x)=\lambda y(x),\ y(0)=0, \ y(a)=1
, \ \lambda= -k^2 < 0.
\end{equation}
The solution of the problem (\ref{eq4}) is given by $$y(x)  = Pe^{-kx} +Qe^{kx}. $$ We can determine the value of constant P and Q by using the given condition $ y(0)=0,  y(a)=1$ with the continuity of the solution. Actually we can divide the solution of the problem (\ref{eq4}) in to two different regions.\\
\[
\begin{cases}
   y_1(x)= \frac{\sinh kx}{\sinh ka},& \text{if } 0\leq x\leq a\\
  y_2(x)=  e^{k(a-x)},& \text{if } a\leq x.
\end{cases}
\]
We integrated the Schr\"odinger equation with respect to $x$ over a small interval $\Delta \epsilon$. 
$$\int_{a-\epsilon}^{a+\epsilon}(-y^{''}-\beta \delta(x-a) y \ )dx= \int_{a-\epsilon}^{a+\epsilon}(\lambda y)dx.$$
 The integral of the second derivative is just the first derivative function and the integral over the function in the right side goes to zero, since it is a continuous, single valued function.  We get,
 $$ -y^{'}|_{a-\epsilon}^{a+\epsilon}-\beta y(a)=0.$$
When $\epsilon \xrightarrow{}0,$
 $ k+k \coth ka =\beta.$
which gives, 
$ \frac{k}{\beta}=\frac{1}{1+ \coth ka}$  and $  \frac{ka}{\beta a}=\frac{1}{1+ \coth ka}.$
Let $ka=A, \beta a= B,$ then $e^{-2A}=1-\frac{2A}{B}.$
Again let $2A=z, $ then we have, 
\begin{equation}{\label{eqz}}
 e^{-z}=1-\frac{z}{B}.   
\end{equation}
From equation (\ref{eqz}) $1-\frac{z}{B}\geq 0.$ Hence, we get $\frac{\beta}{2}\geq k ,$ and it follows
$\frac{\beta^2}{4}\geq k^2.$ 
Since $\lambda =-k^2 \geq \frac{-\beta^2}{4}.$ 
Hence, if $\beta=0$ then there is no possibility of having negative eigenvalues so $\beta_{cr}$ must be greater than zero to produce negative eigenvalues. \\


If we consider the Neumann boundary condition then the equation (\ref{eq1}) becomes
\begin{equation}{\label{eq6}}
-y^{''}-\beta \delta(x-a) y(x)=\lambda y(x),\ y^{'}(0)=0, \ y(a)=1
, \ \lambda= -k^2 < 0,k>0.
\end{equation}

As above, we can divide the solution of this problem in to two different regions, i.e region (I) with $0\leq x <a$ and region (II) with $a<x$. 
\[
\begin{cases}
   y_1(x)= \frac{\cosh kx}{\cosh ka},& \text{if } 0\leq x\leq a\\
  y_2(x)=  e^{k(a-x)},& \text{if } a\leq x.
\end{cases}
\]
We will again integrate the Schr\"odinger equation with respect to $x$ over a small interval and take $\epsilon \xrightarrow{}0.$ Then,
$ k+k \tanh ka =\beta.$ 
Which gives, 
$ \frac{k}{\beta}=\frac{1}{1+ \tanh ka} $
and 
$\frac{ka}{\beta a}=\frac{1}{1+ \tanh ka}.$
Let $ka=A,\ \beta a= B$ then $e^{-2A}=\frac{2A}{B}-1.$
Again let $2A=z, $ then  $ e^{-z}=\frac{z}{B}-1.$ 
It tells us that $\frac{z}{B}-1\geq 0.$ From here, we get the following relation $\frac{\beta}{2}\leq k $ which follows 
$\frac{\beta^2}{4}\leq k^2.$ Since $\lambda =-k^2 \leq \frac{-\beta^2}{4}.$ 
Hence, if $\beta=0$ then there is still a possibility of having a negative eigenvalues so $\beta_{cr}=0.$ \\


If we consider the Robin boundary condition then the equation (\ref{eq1}) becomes
\begin{equation}{\label{eq7}}
-y^{''}-\beta \delta(x-a) y(x)=\lambda y(x),\ \frac{dy}{dx}+y\bigg|_{x=0}=0, \ y(a)=1
, \ \lambda= -k^2 < 0.
\end{equation}
We can divide the solution of this problem in to two different region as described below.

\[
\begin{cases}
   y_1(x)= \frac{k\cosh kx-\sinh kx}{k\cosh ka-\sinh ka},& \text{if } 0\leq x\leq a\\
  y_2(x)=  e^{k(a-x)},& \text{if } a\leq x.
\end{cases}
\]
As above, we will integrate the Schr\"odinger equation with respect to $x$ over a small interval $\Delta \epsilon$.
$$\int_{a-\epsilon}^{a+\epsilon}(-y^{''}-\beta \delta(x-a) y \ )dx= \int_{a-\epsilon}^{a+\epsilon}(\lambda y)dx.$$
After solving this integral problem, we come up with the following equation when $ \epsilon \xrightarrow{}0,$ 
$$ k+k \bigg(\frac{k-\coth ka}{k\coth ka-1}\bigg) =\beta.$$
which gives,
$$ \frac{k}{\beta}=\frac{1}{1+ \bigg(\frac{k-\coth ka}{k\coth ka-1}\bigg)} ,$$
and 
$$ \frac{ka}{\beta a}=\frac{1}{1+ \bigg(\frac{k-\coth ka}{k\coth ka-1}\bigg)}.$$
Let $ka=A, \beta a= B$ then $$e^{-2A}=\frac{2A^2-2Aa-AB+Ba}{AB+Ba}.$$
It tells us that $\frac{2A^2-2Aa-AB+Ba}{AB+Ba}$ must be $\geq 0.$ After solving this inequality, we get the following relation $$\frac{\beta}{2}\leq k, $$ which tells us that $\frac{\beta^2}{4}\leq k^2.$ Since $\lambda =-k^2 \leq \frac{-\beta^2}{4}.$ Hence if $\beta=0$ then there is a possibility of having negative eigenvalues so $\beta_{cr}$ must be zero. 
\end{proof}
Now, It will be shown that the beta critical goes to infinity in the case of Dirichlet boundary condition in the dimension one considering delta potential. We will not consider the Neumann and Robin boundary condition since $\beta_{cr}=0.$
 
\begin{theorem}
Consider the Dirichlet boundary condition where the delta potential is located at $x=a_n$. 
\begin{equation}{\label{eq8}}
-y^{''}-\beta \delta(x-a_n) y(x)=\lambda y(x),\ y(0)=0, \ y(a_n)=1
, \ \lambda= -k^2 < 0,k>0, a_n>0.
\end{equation}
Then, the $\beta_{cr} \xrightarrow{} \infty $ \text{as}\ $ a_n \xrightarrow{}0.$
\end{theorem}
\begin{proof}
We will not have a solution of the equation (\ref{eqz}) if $\frac{1}{B}\geq 1.$ 
That means there is no negative eigenvalues when $ \frac{1}{a}\geq \beta. $
However, if we have the case $\frac{1}{B}\leq 1$ then we would have a solution of the equation (\ref{eqz}). It tells us that when $\frac{1}{ a}\leq \beta,$ we will have the existence of negative eigenvalues. As we defined $\beta_{cr},$ the value of $\beta$ such that equation (\ref{eq4}) does not have negative eigenvalues for $\beta<\beta_{cr}$ and has them if $\beta>\beta_{cr}$, we conclude that $\beta_{cr}=\frac{1}{a}$ for the equation (\ref{eqz}).\\ Similarly, if we consider the case where the delta potential is located at $x=a_n$ where  $ a_n \xrightarrow{}0$  as $n \xrightarrow{}\infty,$ given by (\ref{eq8}) then $$\beta_{cr}=\frac{1}{a_n}.$$
If we approach the potential towards to the boundary which is $a_n \xrightarrow{}0.$ We get,
$\beta_{cr} \xrightarrow{} \infty $ as $a_n \xrightarrow{} 0.$ 
\end{proof}
\section{Dirichlet Boundary condition for $d=2$ with delta potential on the circle} 
 \begin{theorem}\label{thm2}
If we consider the case of Dirichlet Boundary condition for $d=2$ with delta potential on the circle is given by 
\begin{equation}{\label{eq40}}
-\Delta y(x)-\beta \delta_{1+a} y(x)=\lambda y(x),\ y(1)=0, \ y(1+a)=1
, \ \lambda= -k^2 < 0, k>0.
\end{equation}
then the $\beta_{cr}>0$ and $\beta \in (0,\frac{1}{2}).$
\end{theorem} 
\begin{figure}[h!]
    \centering
    \includegraphics[width=0.5\textwidth]{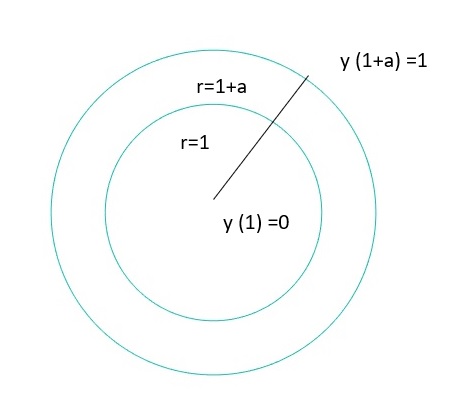}
    \caption{Delta potential in dimension 2.}
    \label{Dellta Potential in d=2.}
\end{figure}
The rotational invariance suggests that the two dimensional Laplacian should take a particularly simple form in polar coordinates. We use polar coordinates $(r, \theta)$ and look for solutions depending only on r. Thus equation (\ref{eq40}) becomes
\begin{equation}{\label{eq10}}
  y^{''}+\frac{y'}{r}-\beta  \delta_{1+a} \ y(r)  =\lambda y(r),\ y(1)=0, \ y(1+a)=1 , \\
   \lambda= -k^2 < 0, k>0.
\end{equation}
\begin{proof}
We will divide the solution of this problem (\ref{eq10}) into two different regions, i.e region (I) with $1\leq r <1+a$ and region (II) with $1+a<r$. 
\[
\begin{cases}
   y_1(r)= \frac{Y_0(k)J_0(kr)-J_0(k)Y_0(kr)}{Y_0(k)J_0(k(1+a))-J_0(k)Y_0(k(1+a))},& \text{if } 1\leq r\leq 1+a\\
  y_2(r)=  \frac{K_0(kr)}{K_0(k(1+a))},& \text{if } 1+a\leq r.
\end{cases}
\]
As above, $$ -y^{'}\bigg|_{1+a-\epsilon}^{1+a+\epsilon}-\beta y(1+a)=0.$$ 
Which gives,
$$ -\big(y_2^{'}|_{1+a+\epsilon}-y_1^{'}|_{1+a-\epsilon}\big)=\beta.$$


When $\epsilon \xrightarrow{}0,$
\begin{equation}{\label{eq125}}
 \frac{-Y_0(k)J_1(k(1+a)+J_0(k)Y_1(k(1+a))}{Y_0(k)J_0(k(1+a))-J_0(k)Y_0(k(1+a))}-\frac{-K_1(k(1+a))}{K_0(k(1+a)}=\frac{\beta}{k}.   
\end{equation}
Where $Y_0$ and $Y_1$ are Bessel function of second kind and and $J_0$ and $J_1$ are Bessel function of first kind. Similarly, $K_0$ and $K_1$ are a modified Bessel function of second kind. Define 
$$ g(k,a)=\frac{g_1(k,a)}{g_2(k,a)}=\frac{-Y_0(k)J_1(k(1+a)+J_0(k)Y_1(k(1+a))}{Y_0(k)J_0(k(1+a))-J_0(k)Y_0(k(1+a))}.$$ It seems that $g(k,a)=-1$ for all the values of $a$  as shown in the figure below.

\begin{figure} [h!]
   \begin{minipage}{0.5\textwidth}
     \centering
\includegraphics[width=1\textwidth]{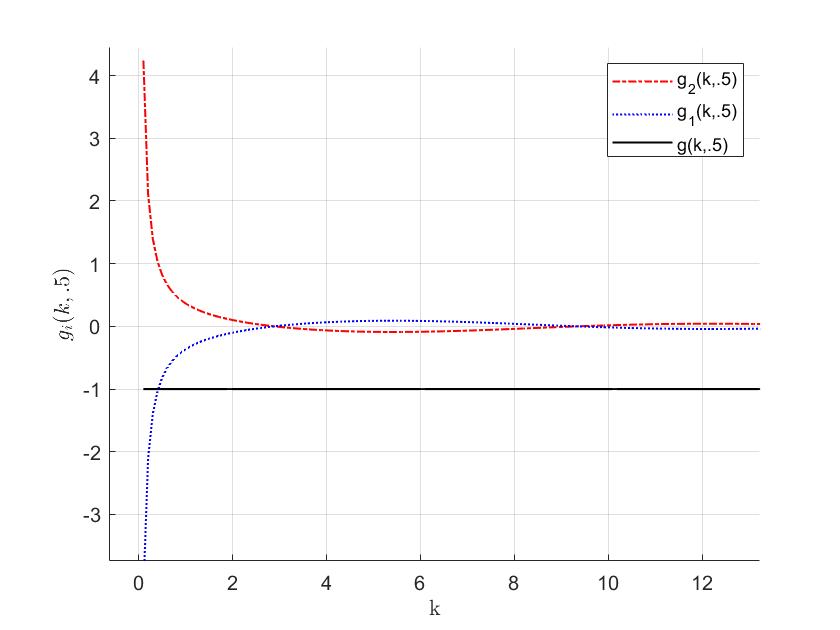}  
\caption{Graph of the function g(k,.5).}
\label{fig:case1} 
   \end{minipage}\hfill
   \begin{minipage}{0.5\textwidth}
     \centering
\includegraphics[width=1\textwidth]{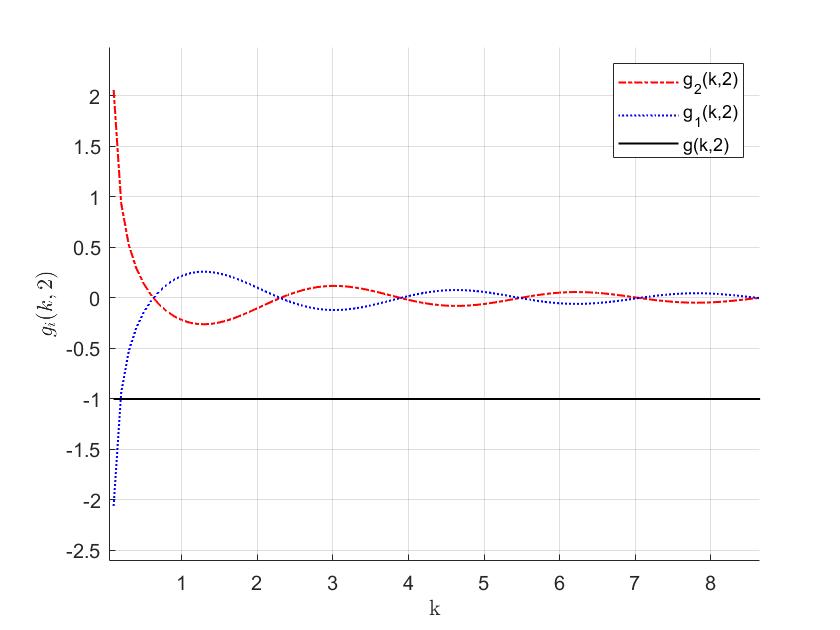}
\caption{Graph of the function g(k,2).}
\label{fig:case1sln} 
   \end{minipage}
\end{figure}

\begin{figure} [h!]
   \begin{minipage}{0.5\textwidth}
     \centering
\includegraphics[width=1\textwidth]{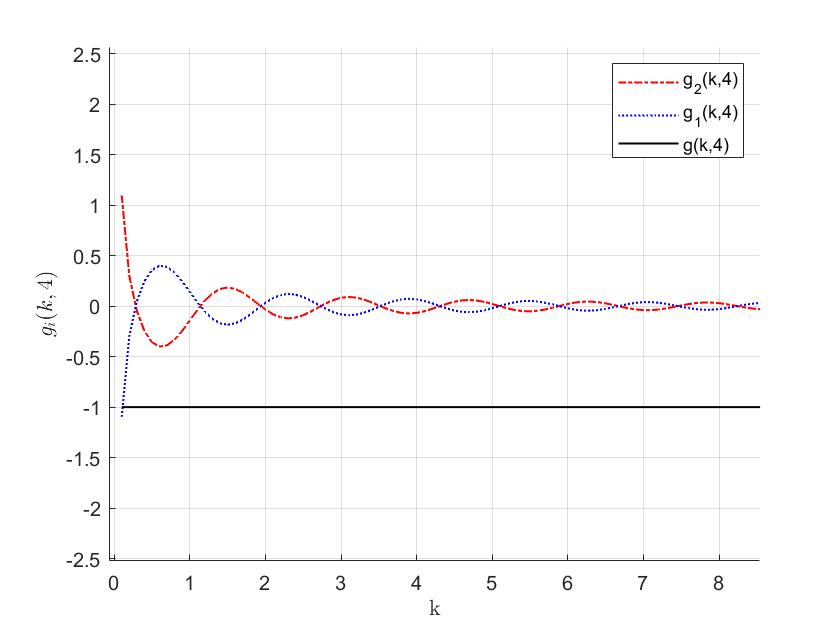}  
\caption{Graph of the function g(k,4).}
\label{fig:case1} 
   \end{minipage}\hfill
   \begin{minipage}{0.5\textwidth}
     \centering
\includegraphics[width=1\textwidth]{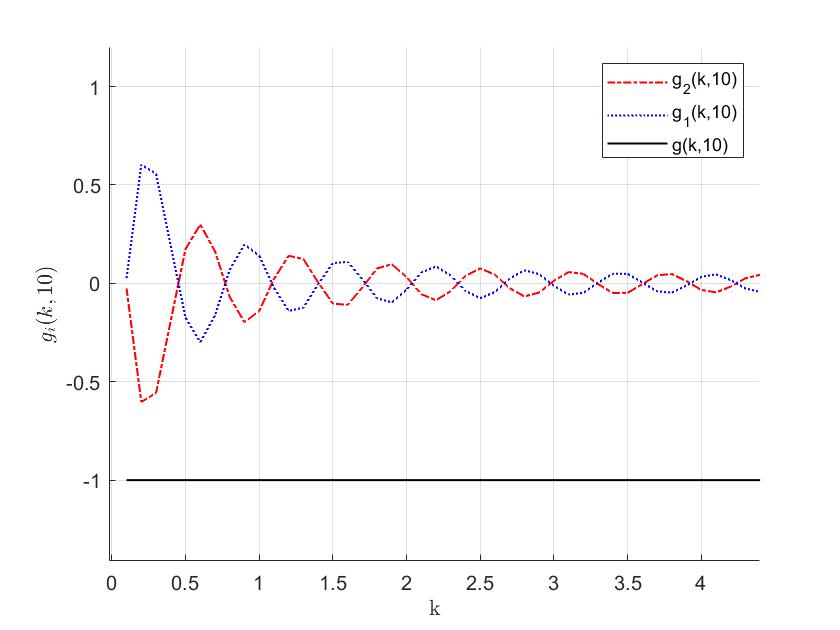}
\caption{Graph of the function g(k,10).}
\label{fig:case1sln} 
   \end{minipage}
\end{figure}

From  equation (\ref{eq125}), we get\\
\begin{equation}{\label{eq126}}
    -1+\frac{K_1(k(1+a))}{K_0(k(1+a)}=\frac{\beta}{k}.
\end{equation}
We will use the following fact from \cite{mbf} to prove that $\beta_{cr}>0.$ 
\begin{lemma}[\cite{mbf}]
Let $p, q\geq 0.$ Then the double inqualities 
\begin{equation}{\label{eq130}}
   1+\frac{1}{2(x+p)}<\frac{K_1(x)}{K_0(x)}<1+\frac{1}{2(x+q)} 
\end{equation}
hold for all $x>0$ if only if $p\geq 1/4$ and $q=0.$
\end{lemma}
 
Now, from equation (\ref{eq126}) and (\ref{eq130}) we get, 
$$\frac{1}{2(x+p)}<\frac{K_1(x)}{K_0(x)}-1<\frac{1}{2(x+q)}.$$
When $x=k(a+1)>0,$
$$\frac{1}{2(k(a+1)+p)}<\frac{K_1(k(a+1))}{K_0(k(a+1))}-1<\frac{1}{2(k(a+1)+q)}.$$
$$\frac{1}{2(k(a+1)+p)}<\frac{\beta}{k}<\frac{1}{2(k(a+1)+q)}.$$
$$\frac{1}{2((a+1)+\frac{p}{k})}< \beta <\frac{1}{2((a+1)+\frac{q}{k})}.$$
Since $a>0, p\geq \frac{1}{4}, k>0$, which tells us that $\beta >0$ and hence $\beta_{cr}>0$ and  $\beta \in (0,\frac{1}{2}).$ 
\end{proof}
\bibliographystyle{plain}
\bibliography{beta_critical}

\end{document}